\documentclass[hidelinks]{siamart251216}

\usepackage{amsmath,amssymb,amsfonts}
\usepackage{graphicx}
\usepackage{textcomp}
\usepackage{url}
\usepackage[final]{microtype}
\usepackage{eso-pic}

\newsiamthm{assumption}{Assumption}
\newsiamthm{example}{Example}
\newsiamremark{remark}{Remark}

\definecolor{twilightcenter}{RGB}{47,20,54}

\graphicspath{figures/}

\headers{Estimating Evolving Functions with Dynamic Gaussian Processes}{J.~van Hulst, W.P.M.H.~Heemels, and D.~Antunes}

\title{Estimating Evolving Functions with Dynamic Gaussian Processes%
  \thanks{The research is carried out as part of the ITEA4 20216 ASIMOV project.
  The ASIMOV activities are supported by the Netherlands Organisation for Applied
  Scientific Research TNO and the Dutch Ministry of Economic Affairs and Climate
  (project number: AI211006). The research leading to these results is partially
  funded by the German Federal Ministry of Education and Research (BMBF) within
  the project ASIMOV-D under grant agreement No.\ 01IS21022G [DLR], based on a
  decision of the German Bundestag.}}

\author{J.S. van Hulst%
  \thanks{Control Systems Technology Section, Department of Mechanical Engineering,
    Eindhoven University of Technology, the Netherlands
    (corresponding author: \email{j.s.v.hulst@tue.nl}).}
\and W.P.M.H. Heemels\footnotemark[2]
\and D.J. Antunes\footnotemark[2]}

\begin{document}

\maketitle

\AddToShipoutPictureFG*{%
    \AtPageUpperLeft{%
        \setlength\unitlength{1in}%
        \hspace*{\dimexpr0.5\paperwidth\relax}%
        \makebox(0,-1.55)[c]{%
            \footnotesize
            \parbox{0.92\paperwidth}{\centering
                J.S. van Hulst \emph{et al.}, ``Estimating Evolving Functions with Dynamic Gaussian Processes.''\\
                This manuscript is a preprint submitted to a SIAM journal.
            }%
        }%
    }%
}%

\AddToShipoutPictureFG*{%
\AtPageLowerLeft{%
\setlength\unitlength{1in}%
\hspace*{\dimexpr0.5\paperwidth\relax}%
\makebox(0,-0.3)[c]{%
\footnotesize
\copyright~The Authors}}}


\begin{abstract}
This paper develops the Dynamic Gaussian Process (DGP), a framework for estimating functions governed by integro-difference equations (IDEs). IDEs model continuous functions that evolve with discrete-time dynamics and arise naturally from time-discretization of linear partial differential equations (PDEs). The DGP extends Gaussian process regression to time-varying functions and extends Kalman filtering to infinite-dimensional states. The DGP posterior remains a Gaussian process with closed-form mean and covariance updates, and separable kernel structure reduces the problem to a finite-dimensional Kalman filter on basis function coefficients. This paper extends the DGP to vector-valued states, enabling the treatment of higher-order PDEs, and provides a stability and approximation error analysis for the basis function approximation.
The functional $L_2$ estimation error decomposes exactly into in-subspace and out-of-subspace contributions, and all approximation errors vanish as the number of basis functions grows. The framework is demonstrated on the heat equation and on the wave equation, the latter with a vector-valued state. Code is available at \url{https://github.com/JvHulst/Dynamic_Gaussian_Processes}.
\end{abstract}


\section{Introduction}
Estimating evolving continuous functions appears in many applications, including heat flow, battery thermal management, fluid dynamics, and population dynamics. Standard estimation methods, however, either address estimation in finite-dimensional spaces (e.g., Kalman filtering) or estimation of static functions (e.g., Gaussian process regression). This paper investigates estimation problems for a class of evolving function models that relaxes both assumptions simultaneously.

The considered class of evolving function models is based on integro-difference equations (IDE), which model continuous functions that evolve with discrete-time dynamics. IDEs are spatio-temporal models with both a time dimension and a space dimension~\cite{Cressie2011,Atluri2019a}. IDE models are currently employed in population evolution modeling~\cite{Lutscher2018}, geographical processes~\cite{Wikle1999,Rozier2007}, and can be connected to the estimation of finite-dimensional time-discretized partial differential equations~\cite[Chapter 6]{Cressie2011}. To estimate the state of such systems, two natural starting points are the Kalman filter (KF)~\cite{Kalman1960}, the quintessential estimator for finite-dimensional linear dynamical systems, and Gaussian process (GP) regression~\cite{Rasmussen2006}, a widely used non-parametric method for estimating static functions. The KF and GP share many similarities, including the assumption of Gaussian noise and the use of Bayesian inference~\cite{Deisenroth2011b,Reece2010a}. In fact, the update step of the KF can be viewed as a GP evaluated on a finite set of points~\cite{Kuper2020}, and conversely, the computational complexity of GP regression can be reduced by incorporating KF logic, both for temporal data~\cite{Hartikainen2010,Corenflos2021} and for spatio-temporal settings~\cite{Sarkka2012,Todescato2020}. This connection has been exploited in the Kriged Kalman Filter (KKF)~\cite{Mardia1998}, which decomposes a spatio-temporal field into basis function coefficients with linear dynamics and a spatially correlated residual estimated via kriging.

In our previous work~\cite{VanHulst2023}, we unified the KF and the GP under a single estimation framework for IDEs, which we coined the Dynamic Gaussian Process (DGP). The key insight is that a GP remains a GP after both IDE evolution and conditioning on observations, so the DGP posterior admits closed-form mean and covariance updates. When the kernels have a separable structure, the infinite-dimensional estimation problem reduces to a finite-dimensional Kalman filter on basis function coefficients. This approximate estimation framework generalizes the results in~\cite{Wikle1999} and shares structural similarities with the KKF; these connections are discussed in Section~\ref{dgp_sec:analysis}.

This paper extends that work in two directions. First, we provide a stability and approximation error analysis for the basis function approximation. The analysis shows that the finite-dimensional transition matrix inherits the spectral properties of the integral operator, that the functional estimation error decomposes exactly into in-subspace and out-of-subspace contributions, and that all approximation errors vanish as $M$ grows. Second, we generalize the DGP to vector-valued states, which enables the treatment of higher-order PDEs through state augmentation.

The rest of this paper is organized as follows. Section~\ref{dgp_sec:problem} formulates the IDE model class, its connection to PDEs, and provides relevant background on the KF and the GP. Section~\ref{dgp_sec:method} presents the exact DGP solution. Section~\ref{dgp_sec:analysis} shows how separable kernel structure reduces the DGP to a finite-dimensional Kalman filter and discusses connections to existing frameworks. Section~\ref{dgp_sec:stability} presents the stability and approximation error analysis, which constitutes the main new contribution beyond~\cite{VanHulst2023}. Section~\ref{dgp_sec:results} presents numerical examples for the heat and wave equations, and Section~\ref{dgp_sec:conclusions} gives conclusions.

\noindent \textbf{Notation.}~Let $\mathbb{R} =(-\infty,\infty)$, and $\mathbb{R}_{\geq0} = [0,\infty)$. Let $\mathbb{N} = \{0,1,\ldots\}$ denote the natural numbers. The identity matrix of size $n$ is denoted by $I_{n}$. The Kronecker product is denoted $\otimes$. Let $\mathbb{S}^n_+:=\{A\in\mathbb{R}^{n\times n} | A \succeq 0 \}$ denote the set of symmetric positive semidefinite matrices of size $n$. A normal distribution with mean vector $\mu \in \mathbb{R}^n$ and covariance matrix $\Sigma \in \mathbb{S}^n_+$ is denoted $\mathcal{N}(\mu,\Sigma)$. A Gaussian process with mean function $\bar{f}(x)$ and covariance function $k(x,x')$ is denoted $\mathcal{GP}(\bar{f}(x),k(x,x'))$.

Given a function $f: \mathcal{X} \to \mathbb{R}$ and a vector $X = [x_1, x_2, \ldots, x_n]^\top \in \mathcal{X}^n$, we use the following compact notation:
\begin{equation*}
    \mathbf{f}(X) := [f(x_1), f(x_2), \ldots, f(x_n)]^\top \in \mathbb{R}^n.
\end{equation*}
Similarly, for a kernel function $k: \mathcal{X} \times \mathcal{X} \to \mathbb{R}$ and vectors $X \in \mathcal{X}^n$, $X' \in \mathcal{X}^{m}$, the kernel matrix $\mathbf{k}(X,X') \in \mathbb{R}^{n \times m}$ is defined entrywise by
\begin{equation*}
    \mathbf{k}(X,X')_{ij} := k(x_i,x'_j),
\end{equation*}
for $1\leq i \leq n$ and $1 \leq j \leq m$. Throughout, boldface denotes evaluation of a function or kernel at a vector of spatial points. When $f: \mathcal{X} \to \mathbb{R}^D$ is vector-valued with components $f^{(1)}, \ldots, f^{(D)}$, the boldface convention extends by stacking output dimensions,
\begin{equation*}
    \mathbf{f}(X) := [\mathbf{f}^{(1)\top}(X), \ldots, \mathbf{f}^{(D)\top}(X)]^\top \in \mathbb{R}^{Dn},
\end{equation*}
and similarly for matrix-valued kernels $K: \mathcal{X} \times \mathcal{X} \to \mathbb{R}^{D \times D}$, where $\mathbf{K}(X,X') \in \mathbb{R}^{Dn \times Dm}$ follows the same dimension-first ordering.

\section{Problem Formulation and Preliminaries}
\label{dgp_sec:problem}
This section introduces the integro-difference equation model class and its connection to partial differential equations, followed by the observation model and the relevant background on Gaussian process regression and Kalman filtering.

\subsection{Integro-Difference Equations}
\label{dgp_subsec:problem}
Consider a function $f_t: \mathcal{X} \to \mathbb{R}^D$ defined on a spatial domain $\mathcal{X}$ that evolves in discrete time according to an integro-difference equation (IDE),
\begin{equation}
\label{dgp_eq:dynamic_function_system}
	f_{t+1}(x) = \int_\mathcal{X} K_f(x,s)\, f_{t}(s)\, d\nu(s) + v_{t}(x),
\end{equation}
with $x \in \mathcal{X}$, $t \in \mathbb{N}$, state function $f_t: \mathcal{X} \to \mathbb{R}^D$, matrix-valued evolution kernel $K_f: \mathcal{X} \times \mathcal{X} \to \mathbb{R}^{D \times D}$, and stochastic disturbance $v_t(x)$. Here, $\nu$ is a $\sigma$-finite reference measure on $\mathcal{X}$, typically the Lebesgue measure for spatially continuous systems or the counting measure when $\mathcal{X}$ is a finite set. IDEs of this form model continuous functions that evolve with discrete-time dynamics. They are spatio-temporal models with many practical applications\cite{Wikle1999,Rozier2007,Lutscher2018,Cressie2011,Atluri2019a}. For $D=1$, the kernel $K_f$ reduces to a scalar function $k_f: \mathcal{X} \times \mathcal{X} \to \mathbb{R}$ and~\eqref{dgp_eq:dynamic_function_system} is the standard scalar IDE.

Let $\mathcal{K}$ denote the integral operator associated with $K_f$, defined by
\begin{equation}
\label{dgp_eq:operator_def}
	(\mathcal{K} f)(x) := \int_\mathcal{X} K_f(x,s)\, f(s)\, d\nu(s),
\end{equation}
so that~\eqref{dgp_eq:dynamic_function_system} can be written compactly as $f_{t+1} = \mathcal{K} f_t + v_t$. We assume that $f_0(x) \sim \mathcal{GP}(\bar{f}_0(x), Q_f(x,x'))$ in which $\bar{f}_0: \mathcal{X} \to \mathbb{R}^D$ and $Q_f: \mathcal{X} \times \mathcal{X} \to \mathbb{R}^{D \times D}$ are the mean function and positive semidefinite matrix-valued covariance function of the initial condition, respectively. The disturbances are $v_t(x) \sim \mathcal{GP}(0,Q_v(x,x'))$, $t \in \mathbb{N}$, with $Q_v: \mathcal{X} \times \mathcal{X} \to \mathbb{R}^{D \times D}$ positive semidefinite and independent of $f_t$.

A natural source of IDE models is the time-discretization of partial differential equations~\cite[Chapter~6]{Cressie2011}. Consider a linear PDE of temporal order $r \geq 1$ on a domain $\mathcal{X} \subseteq \mathbb{R}^d$, subject to boundary conditions on $\partial\mathcal{X}$ such that the initial-boundary value problem is well-posed:
\begin{equation}
\label{dgp_eq:general_pde}
    \frac{\partial^r \phi}{\partial \tau^r}(x,\tau) = \mathcal{L}\, \phi(x,\tau),
\end{equation}
where $\phi: \mathcal{X} \times \mathbb{R}_{\geq 0} \to \mathbb{R}$, $\tau \in \mathbb{R}_{\geq 0}$ is the continuous-time variable, and $\mathcal{L}$ is a linear spatial differential operator. The solution of~\eqref{dgp_eq:general_pde} depends on $r$ initial conditions $\phi(x,0), \partial_\tau\phi(x,0), \ldots, \partial_\tau^{r-1}\phi(x,0)$. Collecting these into a vector-valued state $f(x,\tau) \in \mathbb{R}^D$ with $D \leq r$ through the standard state-space reduction, the time evolution from $0$ to $\tau$ is described by a bounded linear solution operator $T(\tau): L_2(\mathcal{X};\mathbb{R}^D) \to L_2(\mathcal{X};\mathbb{R}^D)$~\cite[Chapter~2]{Pazy1996}:
\begin{equation}
\label{dgp_eq:solution_operator}
    f(\cdot,\tau) = T(\tau)\, f(\cdot,0).
\end{equation}
If the solution operator $T(\tau)$ admits an integral kernel $\mathcal{S}(x,s,\tau) \in \mathbb{R}^{D \times D}$, then~\eqref{dgp_eq:solution_operator} can be written as
\begin{equation}
\label{dgp_eq:general_solution}
    f(x,\tau) = \int_\mathcal{X} \mathcal{S}(x,s,\tau)\, f(s,0)\, ds.
\end{equation}
The kernel $\mathcal{S}$ is the Green's function of the PDE, with the matrix-valued structure arising from the state augmentation when $r > 1$. Existence and regularity of Green's functions for linear PDEs, including the parabolic and hyperbolic classes considered in this paper, is classical~\cite[Chapters~2,~7]{Evans2010}. Section~\ref{dgp_sec:results} derives the kernels in closed form for the heat equation ($D=1$) and the wave equation ($D=2$).

To transition from this continuous dynamic to a discrete-time filtering formulation, we fix a sampling interval $\Delta > 0$. The semigroup property $T(\tau + \Delta) = T(\Delta)\,T(\tau)$ implies $f(x, \tau + \Delta) = \int_\mathcal{X} \mathcal{S}(x,s,\Delta)\,f(s,\tau)\,ds$ for any $\tau \geq 0$. Writing $f_t(x) := f(x,t\Delta)$ for $t \in \mathbb{N}$ then gives
\begin{equation}
\label{dgp_eq:pde_to_ide}
    f_{t+1}(x) = \int_\mathcal{X} \mathcal{S}(x,s,\Delta)\, f_t(s)\, ds,
\end{equation}
which is an instance of~\eqref{dgp_eq:dynamic_function_system} with $K_f(x,s) = \mathcal{S}(x,s,\Delta)$.

\begin{remark}
\label{dgp_rem:pde_consistency}
The IDE~\eqref{dgp_eq:pde_to_ide} is a consistent time discretization of the PDE~\eqref{dgp_eq:general_pde}: setting $v_t = 0$ and letting $\Delta \to 0$,
\begin{equation}
\label{dgp_eq:pde_consistency}
    \frac{(\mathcal{K} f)(x) - f(x)}{\Delta} \to (\mathcal{L} f)(x)
\end{equation}
for all $f$ in the domain of $\mathcal{L}$, recovering the original PDE dynamics. This follows from standard semigroup theory: the solution operators $\{T(\tau)\}_{\tau \geq 0}$ form a strongly continuous ($C_0$-)semigroup whose infinitesimal generator is the spatial differential operator~\cite[Chapter~1]{Pazy1996}. For higher-order PDEs ($r > 1$), the state augmentation produces a first-order system on $L_2(\mathcal{X};\mathbb{R}^D)$, to which the same semigroup argument applies with the generator acting on the product space.
\end{remark}

The PDE connection above provides one natural source of IDE models, but the formulation in~\eqref{dgp_eq:dynamic_function_system} is more general. The evolution kernel $K_f$ can also be specified directly. While the results in this paper hold for general kernel functions $K_f$, $Q_f$, $Q_v$, they will lead to particularly efficient implementations when the kernels take a special separable form. A matrix-valued kernel $K(x,x') \in \mathbb{R}^{D \times D}$ is called separable if each of its entries satisfies
\begin{equation}
\label{dgp_eq:separable_def}
    K_{ij}(x,x')= U^\top(x)\Lambda_{ij} U(x'), \quad i,j \in \{1,\ldots,D\},
\end{equation}
for a shared vector of basis functions
\begin{equation*}
    U(x) := [u_1(x), \ldots, u_M(x)]^\top \in \mathbb{R}^M
\end{equation*}
and coefficient matrices $\Lambda_{ij} \in \mathbb{R}^{M\times M}$. Equivalently, defining $\Lambda := [\Lambda_{ij}]_{i,j=1}^D \in \mathbb{R}^{DM \times DM}$ and writing $\check{U}(x) := I_D \otimes U(x) \in \mathbb{R}^{DM \times D}$ gives the compact form $K(x,x') = \check{U}^\top(x)\, \Lambda\, \check{U}(x')$. For $D=1$, $\check{U}(x) = U(x)$ and this reduces to $k(x,x') = U^\top(x)\Lambda\, U(x')$. The benefits of this structure are further examined in Section~\ref{dgp_sec:analysis}.

\subsection{Observations and Estimation Problem}
\label{dgp_subsec:observations}
At each time step $t$, we obtain $p$ scalar observations $Y_t \in \mathbb{R}^p$ of the function $f_t$ at spatial locations $X_t \in \mathcal{X}^p$, according to
\begin{equation}
\label{dgp_eq:function_observation_model}
	Y_t = \Phi_t\, \mathbf{f}_{t}(X_t) + \mathbf{w}_{t}(X_t),
\end{equation}
where $\mathbf{f}_t(X_t) \in \mathbb{R}^{Dp}$ stacks the state evaluations by output dimension, and $\Phi_t \in \mathbb{R}^{p \times Dp}$ is a known observation matrix that selects which state components are measured. The measurement noise is $w_t(x) \sim \mathcal{GP}(0,Q_w(x,x')),~ t \in \mathbb{N}$, with positive semidefinite $Q_w: \mathcal{X} \times \mathcal{X} \to \mathbb{R}$. For $D=1$, the observation matrix reduces to $\Phi_t = I_p$ and~\eqref{dgp_eq:function_observation_model} recovers the standard observation model $Y_t = \mathbf{f}_t(X_t) + \mathbf{w}_t(X_t)$. For $D>1$, the matrix $\Phi_t$ enables partial state observation; for instance, if only the first of two state components is measured, then $\Phi_t = [I_p ~~ 0_p] \in \mathbb{R}^{p \times 2p}$.

The problem considered in this paper is the estimation of the function $f_N$ using the data set $\{X_{0:N}, Y_{0:N}\}$ in which $X_{0:N} := [X_0, X_1, \ldots, X_N]^\top \in \mathcal{X}^{(N+1) \times p}$ and corresponding $Y_{0:N} := [Y_0, Y_1, \ldots, Y_N]^\top \in \mathbb{R}^{(N+1) \times p}$ as in~\eqref{dgp_eq:function_observation_model}, and $N \in \mathbb{N}$ is arbitrary. Besides enabling the estimation of time-discretized PDEs of the form~\eqref{dgp_eq:general_pde}, this problem can be motivated from two additional perspectives:
\begin{enumerate}
    \item as an extension of GP-based estimation to the case where the Gaussian process evolves in time;
    \item as an extension of Kalman filtering to infinite-dimensional systems described by integro-difference models.
\end{enumerate}
These connections are addressed in the next subsections, which also provide preliminaries for the remainder of the paper.

\subsection{Gaussian Process Regression}
\label{dgp_subsec:GP}
A Gaussian process (GP) is a distribution over the space of functions, fully specified by a prior mean function $\bar{f}(x) = \mathbb{E}[f(x)]$ and a covariance (kernel) function $k(x,x') = \text{cov}(f(x),f(x'))$~\cite{Rasmussen2006}. Given noisy observations $y_t = f(x_t) + w_t$ with $w_t \sim \mathcal{N}(0,\sigma^2)$ at inputs $X = \{x_0, \ldots, x_N\}$, the posterior distribution is again Gaussian:
\begin{equation}
\label{dgp_eq:GP_posterior}
\begin{aligned}
	{\hat{f}}(x) &= \bar{f}(x) + L(x,X)(Y-\mathbf{\bar{f}}(X)), \\
	{\hat{c}}(x,x') &= {k}(x,x') - L(x,X) \mathbf{k}(X,x'),
\end{aligned}
\end{equation}
where $L(x,X) := \mathbf{k}(x,X) \left[\mathbf{k}(X,X)+\sigma^2 I_N\right]^{-1}$. The computation of the posterior scales with $O(N^3)$ due to the matrix inverse, which can be reduced through sparse methods~\cite{Lee2020}, recursive methods~\cite{Huber2014}, or basis function approximations~\cite{Cressie2008}.

The standard GP framework is a special case of~\eqref{dgp_eq:dynamic_function_system} when
\begin{equation*}
    f_{t+1}(x) = f_{t}(x)
\end{equation*}
with $f_{0}(x) = f(x)$. These dynamics are obtained by choosing $D=1$, $\mathcal{K} = \mathcal{I}$ (the identity operator), and $Q_v(x,x') = 0$. The estimation problem in Section~\ref{dgp_subsec:observations} can therefore be seen as an extension of the GP framework to the case where the unknown function $f$ evolves according to~\eqref{dgp_eq:dynamic_function_system}.

\subsection{Kalman Filtering}
\label{dgp_subsec:KF}
The Kalman filter is a well-known algorithm to estimate the state of a dynamical system described by
\begin{align}
\label{dgp_eq:dynamic_state_system}
	x_{t+1} &= A x_t + v_t, \\
\label{dgp_eq:state_observation_model}
	y_{t} &= C_t x_t + w_t,
\end{align}
where $x_t \in \mathbb{R}^{n}$ the state to be estimated, $y_t \in \mathbb{R}^{p}$ the system output, $v_t \sim \mathcal{N}(0,V)$, $t \in \mathbb{N}$  the disturbance on the state with covariance matrix $V \in \mathbb{S}^n_+$, and $w_t \sim \mathcal{N}(0,W_t),~ t \in \mathbb{N},$ the measurement noise with covariance matrix $W_t \in \mathbb{S}^p_+$, with $x_0 \sim \mathcal{N}(\bar{x}_0,\bar{S})$. The output matrix $C_t$ and the noise covariance $W_t$ are chosen time-varying because in the IDE setting of Section~\ref{dgp_subsec:problem}, the observation locations $X_t$ may differ at each time step.

Important for what follows is to detail the two steps that together constitute the Kalman filter, called the update and the prediction steps. The equations for the update step are given by
\begin{equation}
\label{dgp_eq:Kalman_filter_update}
\begin{aligned}
	\hat{x}_{t|t}&=\hat{x}_{t|t-1} + L_{t} (y_t-C_t \hat{x}_{t|t-1}), \\
	S_{t|t} &= S_{t|t-1} - L_{t} C_t S_{t|t-1},
\end{aligned}
\end{equation}
with $L_{t} = S_{t|t-1} C_t^\top(C_t S_{t|t-1}C_t^\top+W_t)^{-1}$ the Kalman gain, $\hat{x}_{0|-1}=\bar{x}_0$ and $S_{0|-1}=\bar{S}$. Here, $\hat{x}_{t|l}$ denotes the estimate of $x_t$ at time step $l$, and $S_{t|l}$ denotes the estimate of the state covariance $\text{cov}(x_t,x_t)$ at time step $l$. The equations for the prediction step in the KF are given by
\begin{equation}
\label{dgp_eq:Kalman_filter_prediction}
\begin{aligned}
	\hat{x}_{t+1|t} &= A \hat{x}_{t|t}, \\
	S_{t+1|t} &= A S_{t|t} A^\top + V.
\end{aligned}
\end{equation}

Note that the linear and quadratic forms of, respectively, the mean estimate and the covariance estimate are closed under the operation of both the update and the prediction steps. Furthermore, these steps rely only on the previous estimates and the current data. Hence, the Kalman filter can be implemented recursively. Under the presented assumptions of linear dynamics, and zero-mean Gaussian noise and disturbances, the Kalman filter is the optimal estimator for this problem~\cite{Kalman1960}.

The standard Kalman filter can be seen as a special case of the estimation problem considered in this paper. Specifically, let $\mathcal{X} = \{x_1, x_2, \ldots, x_m\}$ be a finite set equipped with the counting measure, and let the state to be estimated be $\mathbf{f}_t(X) \in \mathbb{R}^{Dm}$. Under the counting measure, the integral in~\eqref{dgp_eq:dynamic_function_system} reduces to a finite sum, and the system becomes a standard linear state-space model with transition matrix $A \in \mathbb{R}^{Dm \times Dm}$ whose entries are determined by the kernel evaluations $K_f(x_i, x_j) \in \mathbb{R}^{D \times D}$.

The next section details the extension of the Kalman filter to infinite-dimensional systems modeled by IDEs, which presents the solution to the problem posed in Section~\ref{dgp_subsec:problem}.

\section{Dynamic Gaussian Processes}
\label{dgp_sec:method}
Here, we consider the estimation problem of the system~\eqref{dgp_eq:dynamic_function_system}, which we refer to as the dynamic Gaussian process (DGP). The key results are two theorems showing that if the initial condition of the IDE in~\eqref{dgp_eq:dynamic_function_system} is a GP, it remains a GP after evolving and after incorporating observations.
\begin{theorem}
\label{dgp_thm:gp_ide_prediction}
    Suppose that $f_0(x) \sim \mathcal{GP}(\bar{f}_0(x),Q_f(x,x'))$ with $\bar{f}_0: \mathcal{X} \to \mathbb{R}^D$ and $Q_f: \mathcal{X} \times \mathcal{X} \to \mathbb{R}^{D \times D}$, and that $v_0(x) \sim \mathcal{GP}(0,Q_v(x,x'))$ with $Q_v: \mathcal{X} \times \mathcal{X} \to \mathbb{R}^{D \times D}$, independent of $f_0$. Consider
    \begin{equation}
        f_1(x) = \int_\mathcal{X} K_f(x,s) f_0(s)\, d\nu(s) + v_0(x),
    \end{equation}
    with $K_f: \mathcal{X} \times \mathcal{X} \to \mathbb{R}^{D \times D}$. Then $f_1(x)$ is a GP with mean $\bar{f}_1: \mathcal{X} \to \mathbb{R}^D$ and covariance $Q_1: \mathcal{X} \times \mathcal{X} \to \mathbb{R}^{D \times D}$ given by
    \begin{align}
    \label{dgp_eq:thm1_mean}
        \bar{f}_1(x) &= \int_\mathcal{X} K_f(x,s)\bar{f}_0(s)\, d\nu(s),\\
    \label{dgp_eq:thm1_cov}
        Q_1(x,x') &= \int_\mathcal{X}\!\int_\mathcal{X} K_f(x,s)\,Q_f(s,s')\,K_f^\top(x',s')d\nu(s)\, d\nu(s')
        + Q_v(x,x').
    \end{align}
\end{theorem}
\begin{proof}
    The mean of $f_1(x)$ follows from linearity of expectation and the integral:
    \begin{equation*}
    \begin{aligned}
        \mathbb{E}[f_1(x)] &= \int_\mathcal{X} K_f(x,s)\,\mathbb{E}[f_0(s)]\,d\nu(s) + \mathbb{E}[v_0(x)] \\
        &= \int_\mathcal{X} K_f(x,s)\,\bar{f}_0(s)\,d\nu(s).
    \end{aligned}
    \end{equation*}
    For the covariance, write the centered variable as
    \begin{equation*}
        f_1(x) - \bar{f}_1(x) = \int_\mathcal{X} K_f(x,s)\bigl(f_0(s) - \bar{f}_0(s)\bigr)\,d\nu(s) + v_0(x).
    \end{equation*}
    Expanding $\mathbb{E}\bigl[(f_1(x) - \bar{f}_1(x))(f_1(x') - \bar{f}_1(x'))^\top\bigr]$ produces three terms: the integral--integral product, the noise--noise product $\mathbb{E}[v_0(x)\,v_0^\top(x')] = Q_v(x,x')$, and two cross terms between the integral and $v_0$. Because $f_0$ and $v_0$ are independent, the cross terms vanish. For the integral--integral product, $\mathbb{E}\bigl[(f_0(s)-\bar{f}_0(s))(f_0(s')-\bar{f}_0(s'))^\top\bigr] = Q_f(s,s')$ by definition of the covariance kernel, so exchanging expectation and integration gives
    \begin{equation*}
        \int_\mathcal{X}\!\int_\mathcal{X} K_f(x,s)\,Q_f(s,s')\,K_f^\top(x',s')\,d\nu(s)\,d\nu(s'),
    \end{equation*}
    which is the first term of~\eqref{dgp_eq:thm1_cov}.
    To show that $f_1$ is a GP, note that any finite collection of evaluations $f_1(x_1), \ldots, f_1(x_n)$ is a linear transformation of the jointly Gaussian variables $\{f_0(s): s \in \mathcal{X}\}$ and $\{v_0(x_i)\}_{i=1}^n$, and is therefore jointly Gaussian.
\end{proof}
With the prediction step established, the remaining ingredient is conditioning the GP on observations.

\begin{theorem}
\label{dgp_thm:gp_conditioning}
    Suppose that $f_0(x) \sim \mathcal{GP}(\bar{f}_0(x),Q_f(x,x'))$ with $\bar{f}_0: \mathcal{X} \to \mathbb{R}^D$ and $Q_f: \mathcal{X} \times \mathcal{X} \to \mathbb{R}^{D \times D}$, and that $w_0(x) \sim \mathcal{GP}(0,Q_w(x,x'))$ is independent of $f_0$. Consider the observations
    \begin{equation}
        Y_0 = \Phi_0\,\mathbf{f}_0(X_0) + \mathbf{w}_0(X_0),
    \end{equation}
    with $\Phi_0 \in \mathbb{R}^{p \times Dp}$ the observation matrix. Conditioning the GP $f_0$ on $Y_0$ yields another GP with mean
    \begin{equation}
    \label{dgp_eq:thm2_mean}
        \bar{f}_0(x) + \mathbf{Q_f}(x,X_0) \Phi_0^\top L_0 (Y_0 - \Phi_0\,\mathbf{\bar{f}}_0(X_0)),
    \end{equation}
    and covariance
    \begin{equation}
    \label{dgp_eq:thm2_cov}
        Q_f(x,x') - \mathbf{Q_f}(x,X_0) \Phi_0^\top L_0\, \Phi_0\, \mathbf{Q_f}(X_0,x'),
    \end{equation}
    in which $L_0 := \left[\Phi_0\,\mathbf{Q_f}(X_0,X_0)\,\Phi_0^\top + \mathbf{Q_w}(X_0,X_0)\right]^{-1}$, and $\mathbf{Q_f}(x,X_0) \in \mathbb{R}^{D \times Dp}$ denotes the cross-covariance between $f_0(x)$ and the stacked evaluations $\mathbf{f}_0(X_0)$.
\end{theorem}
\begin{proof}
    Since $f_0$ is a GP and $w_0$ is an independent GP, the observations $Y_0 = \Phi_0\,\mathbf{f}_0(X_0) + \mathbf{w}_0(X_0)$ are an affine transformation of jointly Gaussian variables. The joint distribution of $Y_0$ and $f_0(x)$ at an arbitrary $x \in \mathcal{X}$ is therefore Gaussian:
    \begin{equation*}
    \begin{bmatrix} Y_0 \\ f_0(x) \end{bmatrix}
    \sim \mathcal{N}\!\left(
    \begin{bmatrix} \Phi_0\,\mathbf{\bar{f}}_0(X_0) \\ \bar{f}_0(x) \end{bmatrix},
    \begin{bmatrix}
        \Sigma_0 & \Phi_0\,\mathbf{Q_f}(X_0,x) \\
        \mathbf{Q_f}(x,X_0)\,\Phi_0^\top & Q_f(x,x)
    \end{bmatrix}
    \right),
    \end{equation*}
    with $\Sigma_0 := \Phi_0\,\mathbf{Q_f}(X_0,X_0)\,\Phi_0^\top + \mathbf{Q_w}(X_0,X_0)$ the covariance of $Y_0$. Applying the standard Gaussian conditioning formula $\mathbb{E}[a \mid b] = \mu_a + \Sigma_{ab}\Sigma_{bb}^{-1}(b - \mu_b)$, $\text{cov}(a \mid b) = \Sigma_{aa} - \Sigma_{ab}\Sigma_{bb}^{-1}\Sigma_{ba}$ to this joint distribution, with $L_0 := \Sigma_0^{-1}$, gives the stated posterior mean and covariance. Since $x$ is arbitrary, the conditional distribution defines a GP.
\end{proof}
Alternating the prediction step of Theorem~\ref{dgp_thm:gp_ide_prediction} and the update step of Theorem~\ref{dgp_thm:gp_conditioning} yields the DGP estimator.

\subsection{DGP Estimator}
We aim to estimate the evolving function $f_{t}(x) \in \mathbb{R}^D$ using the model~\eqref{dgp_eq:dynamic_function_system} and observation model~\eqref{dgp_eq:function_observation_model}. The estimate is a GP characterized by a mean function $\hat{f}_{t|l}: \mathcal{X} \to \mathbb{R}^D$ and a covariance function $\hat{c}_{t|l}: \mathcal{X} \times \mathcal{X} \to \mathbb{R}^{D \times D}$. The double subscripts indicate that we estimate $f_t$ using information up to time step $l$. Prior knowledge is embedded through the initial conditions $\hat{f}_{0|-1} = \bar{f}_0$ and $\hat{c}_{0|-1} = Q_f$. The update step, based on noisy measurements $Y_t$ at spatial locations $X_t$, is given by:
\begin{equation}
\label{dgp_eq:DGP_update}
\begin{aligned}
	{\hat{f}}_{t|t}(x) &= {\hat{f}}_{t|t-1}(x) + L_t(x,X_t) (Y_t-\Phi_t\,\mathbf{\hat{f}}_{t|t-1}(X_t)),\\
	 {\hat{c}}_{t|t}(x,x') &= {\hat{c}}_{t|t-1}(x,x') - L_t(x,X_t)\, \Phi_t\, \mathbf{\hat{c}}_{t|t-1}(X_t,x') ,
\end{aligned}
\end{equation}
in which
\begin{equation*}
\begin{aligned}
    L_t(x,X_t) 
    := \mathbf{\hat{c}}_{t|t-1}(x,X_t)\, \Phi_t^\top 
    \left[\Phi_t\,\mathbf{\hat{c}}_{t|t-1}(X_t,X_t)\, \Phi_t^\top + \mathbf{Q_w}(X_t,X_t) \right]^{-1}.
\end{aligned}
\end{equation*}

The prediction step is given by
\begin{equation}
\label{dgp_eq:DGP_prediction}
\begin{aligned}
	{\hat{f}}_{t+1|t}(x) &= \int_\mathcal{X} {K_f}(x,s)\,{\hat{f}}_{t|t}(s)\, d\nu(s) ,\\
	{\hat{c}}_{t+1|t}(x,x') &= \int_\mathcal{X}\!\int_\mathcal{X} {K_f}(x,s)\, \hat{c}_{t|t}(s,s')\, {K_f}^\top(x',s')d\nu(s)\, d\nu(s')
    + {Q_v}(x,x').
\end{aligned}
\end{equation}

While the DGP updates~\eqref{dgp_eq:DGP_update}--\eqref{dgp_eq:DGP_prediction} are exact, their direct implementation requires integrals over the full domain at every prediction step, and these integrals rarely admit analytical solutions. Moreover, the posterior mean and covariance expressions grow more complicated at each time step. The next section shows that this cost can be managed when the kernels have a separable structure.

\section{Separable Kernels}
\label{dgp_sec:analysis}
In this section, we show that when $K_f$, $Q_f$, and $Q_v$ are separable kernels we can perform the required computations of the previous section efficiently. We start with the DGP with truly separable kernels in Section~\ref{dgp_subsec:seperable_DGP}, and show in Section~\ref{dgp_subsec:approximate_DGP} that separable kernels can be used to approximate the problem from Section~\ref{dgp_subsec:problem}.
\subsection{DGP with Separable Kernels}
\label{dgp_subsec:seperable_DGP}
To illustrate the reduced computational cost that follows from the separability of the kernels $K_f$, $Q_f$, and $Q_v$, consider the following lemmas.
\begin{lemma}
\label{dgp_lem:seperable_kernels}
Suppose that
\begin{equation}
\label{dgp_eq:seperable_assumption}
    \begin{aligned}
        K_f(x,x') &= \check{U}^\top(x)\, \Lambda\, \check{U}(x'),\\
        Q_f(x,x') &= \check{U}^\top(x)\, \Lambda_f\, \check{U}(x'),\\
        Q_v(x,x') &= \check{U}^\top(x)\, \Lambda_v\, \check{U}(x')
    \end{aligned}
\end{equation}
with $\Lambda \in \mathbb{R}^{DM \times DM}$, $\Lambda_f, \Lambda_v \in \mathbb{S}_+^{DM}$ and
\begin{equation*}
    U(x) := [u_1(x), u_2(x), \ldots, u_M(x)]^\top
\end{equation*}
a vector of basis functions in which $u_i: \mathcal{X} \to \mathbb{R}$, for $i \in \{1,2,\ldots,M\}$. Define the transition matrix, observation matrix, and measurement noise covariance
\begin{equation}
\label{dgp_eq:KF_matrices}
\begin{aligned}
    A_M &:= \Lambda\,(I_D \otimes \Lambda_U) \in \mathbb{R}^{DM \times DM}, \\
    C_t &:= \Phi_t\,\check{\mathbf{U}}^\top(X_t) \in \mathbb{R}^{p \times DM}, \\
    W_t &:= \mathbf{Q_w}(X_t, X_t) \in \mathbb{R}^{p \times p},
\end{aligned}
\end{equation}
where $\check{\mathbf{U}}(X_t) := I_D \otimes \mathbf{U}(X_t) \in \mathbb{R}^{Dp \times DM}$ denotes the block-diagonal basis evaluation at the observation locations $X_t$ and $\Lambda_U := \int_\mathcal{X} U(x)U^\top(x)\, d\nu(x) \in \mathbb{R}^{M \times M}$ is the Gram matrix of the basis. Then the update step in~\eqref{dgp_eq:DGP_update} reduces to
\begin{equation}
\label{dgp_eq:reduced_DGP_covariance_update}
    {\hat{c}}_{t|t}(x,x') = \check{U}^\top(x)\, \Psi_{t|t}\, \check{U}(x')
\end{equation}
with
\begin{equation*}
    \Psi_{t|t} := \Psi_{t|t-1} - \Gamma_t\, C_t\, \Psi_{t|t-1}
\end{equation*}
in which $\Psi_{0|-1}=\Lambda_f$, and
\begin{equation*}
    \Gamma_t := \Psi_{t|t-1}\, C_t^\top \left[C_t\, \Psi_{t|t-1}\, C_t^\top + W_t\right]^{-1}.
\end{equation*}

Additionally, the prediction step in~\eqref{dgp_eq:DGP_prediction} is reduced to
\begin{equation}
\label{dgp_eq:reduced_DGP_covariance_prediction}
    {\hat{c}}_{t+1|t}(x,x') = \check{U}^\top(x)\,\Psi_{t+1|t}\,\check{U}(x')
\end{equation}
with
\begin{equation*}
    \Psi_{t+1|t} := A_M\, \Psi_{t|t}\, A_M^\top + \Lambda_v.
\end{equation*}
\end{lemma}
\begin{proof}
Initialize $\hat{c}_{0|-1}(x,x') = \check{U}^\top(x)\, \Psi_{0|-1}\, \check{U}(x')$ with $\Psi_{0|-1} = \Lambda_f$, then substitute~\eqref{dgp_eq:seperable_assumption} into~\eqref{dgp_eq:DGP_update} and~\eqref{dgp_eq:DGP_prediction}. The key identity is $\check{U}(s)\,\check{U}^\top(s) = I_D \otimes U(s)U^\top(s)$, which after integration yields $I_D \otimes \Lambda_U$.
\end{proof}
\begin{remark}
    Note that $Q_w$ does not need to be approximated, as the measurement noise function $w_t(x)$ is always evaluated for a finite set of spatial locations $X_t$ when taking observations according to~\eqref{dgp_eq:function_observation_model}.
\end{remark}
Under assumption~\eqref{dgp_eq:seperable_assumption}, the covariance estimate is closed under both the update and the prediction steps. However, the mean estimate computation still scales poorly as the time index $t$ increases. Efficient implementation of the mean estimate is enabled by the following lemma.
\begin{lemma}
\label{dgp_lem:prior_mean}
    If~\eqref{dgp_eq:seperable_assumption} holds and additionally
    \begin{equation}
        \label{dgp_eq:DGP_prior_mean_assumption}
        \bar{f}_0(x) = \check{U}^\top(x)\,\bar{z},
    \end{equation}
    with $\bar{z} \in \mathbb{R}^{DM}$ and $U(x)$ the same vector of functions as in~\eqref{dgp_eq:seperable_assumption}, then the update of the mean estimate reduces to
    \begin{equation}
    \label{dgp_eq:reduced_DGP_mean_update}
        {\hat{f}}_{t|t}(x) = \check{U}^\top(x)\, z_{t|t},
    \end{equation}
    with
    \begin{equation*}
        z_{t|t} := z_{t|t-1} + \Gamma_t(Y_t - C_t\, z_{t|t-1}).
    \end{equation*}

    Additionally, the prediction step in~\eqref{dgp_eq:DGP_prediction} is reduced to
    \begin{equation}
    \label{dgp_eq:DGP_mean_prediction}
        {\hat{f}}_{t+1|t}(x) = \check{U}^\top(x)\, z_{t+1|t},
    \end{equation}
    with
    \begin{equation*}
        z_{t+1|t} = A_M\, z_{t|t}.
    \end{equation*}
\end{lemma}
\begin{proof}
    Substitute~\eqref{dgp_eq:DGP_prior_mean_assumption} into~\eqref{dgp_eq:DGP_update} and~\eqref{dgp_eq:DGP_prediction}.
\end{proof}

Note the similarity of the expressions in \textit{Lemma}~\ref{dgp_lem:prior_mean} to the standard KF. With \textit{Lemma~\ref{dgp_lem:seperable_kernels}} and \textit{Lemma~\ref{dgp_lem:prior_mean}}, the DGP can be interpreted as a Kalman filter where the state vector $z_t \in \mathbb{R}^{DM}$ represents the coefficients of the function estimate projected onto the basis $U$. The matrices $A_M$, $C_t$, $\Lambda_v$, and $W_t$ correspond directly to $A$, $C_t$, $V$, and $W_t$ from Section~\ref{dgp_subsec:KF}.

If both lemmas apply, the problem presented in Section~\ref{dgp_sec:problem} can be exactly estimated using the equations presented in this section. For the more general case, where~\eqref{dgp_eq:dynamic_function_system} does not satisfy the lemmas, the problem can be approximated using the basis function methods detailed next.

\subsection{Approximate DGP}
\label{dgp_subsec:approximate_DGP}
In general, the function $f_t$ does not lie exactly in $\operatorname{span}\{u_1, \ldots, u_M\}$. To handle this case, we define the coefficient representation of an arbitrary $f_t \in L_2(\mathcal{X})$ as the vector $z_t \in \mathbb{R}^{DM}$ satisfying $\check{U}^\top(x) z_t = (\pi_M f_t)(x)$, where $\pi_M$ denotes the orthogonal projection onto $\operatorname{span}\{u_1, \ldots, u_M\}$. This generalizes the coefficient vector from Lemma~\ref{dgp_lem:prior_mean}, where $f_t$ was assumed to lie exactly in the span of $U$.

The kernel functions $K_f$, $Q_f$, $Q_v$, and the prior mean $\bar{f}_0$ can be approximated in the separable form~\eqref{dgp_eq:seperable_assumption}--\eqref{dgp_eq:DGP_prior_mean_assumption} by $\mathcal{L}_2$-projection onto the basis $U$. This approximation relates to reduced-model Kalman filtering~\cite{Farrell2001}, but with an infinite-dimensional full model. The resulting coefficient vector for the prior mean is
\begin{equation}
\label{dgp_eq:closed-form_least_squares_mean}
    \bar{z}^* = (I_D \otimes \Lambda_U^{-1}) \int_\mathcal{X}\check{U}(x)\, \bar{f}_0(x)\, d\nu(x),
\end{equation}
and the $DM \times DM$ coefficient matrix for a matrix-valued kernel $K \in \{K_f, Q_f, Q_v\}$ is
\begin{equation}
\label{dgp_eq:closed-form_least_squares_kernel}
    \Lambda^*_K = (I_D \otimes \Lambda_U^{-1})\, J_K\, (I_D \otimes \Lambda_U^{-1}),
\end{equation}
where
\begin{equation*}
    J_K := \int_\mathcal{X}\!\!\int_\mathcal{X} \check{U}(x)\, K(x,x')\, \check{U}^\top(x')\, d\nu(x)\, d\nu(x').
\end{equation*}
In practice, the required integrals can be computed approximately using Riemann sums, which only requires pointwise evaluations of $K_f$, $Q_f$, $Q_v$, $\bar{f}_0$.

For the minimizer to be unique, the basis functions must be linearly independent. If additionally the basis is orthonormal, i.e., $\langle u_i, u_j \rangle := \int_\mathcal{X} u_i(x) u_j(x)\, d\nu(x) = \delta(i-j)$, then $\Lambda_U = I_M$ and the projection decouples into independent scalar minimizations.

Common choices for $U$ are the Fourier basis, piecewise-constant (bin) functions, radial basis functions, and hat functions. The approximation error introduced by the projection can be made arbitrarily small by increasing $M$, under mild conditions on the basis functions and the kernels. This is further analyzed in Section~\ref{dgp_sec:stability}.

\subsection{Connection to Existing Frameworks}
The approximate DGP presented in this section is closely related to the work of~\cite{Wikle1999}, who also project the IDE estimation problem onto a set of basis functions. In that work, the basis $U$ is required to be orthonormal, the initial condition $f_0$ is assumed to lie within the span of $U$, and the evolution kernel $K_f$ is assumed to have the separable structure $K_f(x,x') = \tilde{U}^\top(x) U(x')$ for some unknown function vector $\tilde{U}(x)$. The observations $Y_t$ are assumed to occur at the same spatial locations at every time step, with the number of measurement locations satisfying $p \geq M$. Under these assumptions, the coefficient dynamics reduce to a finite-dimensional state-space model that can be estimated using a standard Kalman filter. In contrast, the approximate DGP imposes none of these restrictions: the basis $U$ need not be orthonormal, the spatial locations $X_t$ may vary in time, and $p$ may be smaller than $M$. The key reason is that the approximate DGP is derived from the exact DGP by projecting the posterior onto the basis, rather than by directly assuming a finite-dimensional model.

The approximate DGP also shares structural similarities with the Kriged Kalman Filter (KKF)~\cite{Mardia1998}. In the KKF, a spatio-temporal field is decomposed into basis function coefficients that evolve with linear dynamics, and a spatially correlated residual that is estimated via kriging. Both the approximate DGP and the KKF propagate a finite-dimensional state vector $z_t \in \mathbb{R}^{DM}$ with a transition matrix and apply Kalman filter updates. However, the KKF transition matrix is specified directly, without connection to an evolution kernel $K_f$ or to an underlying exact infinite-dimensional estimation problem. In the DGP, the transition matrix $A_M = \Lambda\,(I_D \otimes \Lambda_U)$ is derived from the evolution kernel through the $\mathcal{L}_2$-projection, which provides a principled link between the finite-dimensional approximation and the exact problem.

A further distinction is that the IDE model in~\eqref{dgp_eq:dynamic_function_system} supports vector-valued states ($D > 1$), which enables the treatment of higher-order PDEs through state augmentation as described in Section~\ref{dgp_subsec:problem}. The DGP estimator inherits this generality, whereas the frameworks of~\cite{Wikle1999} and~\cite{Mardia1998} are restricted to scalar states. What remains is to quantify the error introduced by the basis function approximation.

\section{Stability and Estimation Error}
\label{dgp_sec:stability}
The approximate DGP from Section~\ref{dgp_subsec:approximate_DGP} replaces the infinite-dimensional estimation problem with a Kalman filter on the coefficient vector $z_t \in \mathbb{R}^{DM}$, with transition matrix $A_M$, observation matrix $C_t$, and noise covariance $W_t$ as defined in~\eqref{dgp_eq:KF_matrices}. For the steady-state analysis in this section, we assume stationary measurement locations $X_t = X$ for all $t$, so that $C_t = C$ and $W_t = W$. This assumption is standard in Kalman filter convergence theory~\cite{Anderson2005} and simplifies the notation, though the pointwise-in-time error identity derived below holds without it.

The analytical challenge is that the filter operates on a finite-dimensional subspace while the true state lives in $L_2(\mathcal{X};\mathbb{R}^D)$. When $f_t$ does not lie in $\operatorname{span}\{u_1, \ldots, u_M\}$, the evolution kernel couples the modelled and unmodelled components, producing a leakage term in the coefficient error dynamics that the filter ignores. The filter covariance $\Psi_{t|t}$ consequently underestimates the true estimation error. This section derives an exact decomposition of the functional estimation error that separates the noise-induced and truncation-induced contributions, and shows that all approximation errors vanish as $M \to \infty$.

\subsection{Stability of the Estimation Loop}
\label{dgp_subsec:filter_stability}
If the evolution kernel $K_f$ is square-integrable, then $\mathcal{K}$ is a Hilbert--Schmidt operator and therefore compact~\cite[Theorem~2.22]{Kress1999}. Its spectrum consists of at most a countable set of eigenvalues with the only accumulation point at zero~\cite[Theorem~3.7]{Kress1999}.
\begin{definition}
\label{dgp_def:stability}
The state dynamics~\eqref{dgp_eq:dynamic_function_system} are called \emph{stable} if the spectral radius of the integral operator satisfies $\rho(\mathcal{K}) := \sup\{|\lambda| : \lambda \in \sigma(\mathcal{K})\} < 1$.
\end{definition}
Let $\pi_M: L_2(\mathcal{X}) \to \operatorname{span}\{u_1, \ldots, u_M\}$ denote the orthogonal projection and $\mathcal{K}_M := \pi_M \mathcal{K}\, \pi_M$ the projected operator, whose matrix representation is $A_M$. The following convergence result is a classical consequence of Galerkin approximation theory for compact operators.
\begin{proposition}[Galerkin convergence, {\cite[Theorems~4.7, 10.9, 10.20]{Kress1999}}]
\label{dgp_prop:galerkin_convergence}
If the basis $\{u_m\}_{m=1}^\infty$ is complete in $L_2(\mathcal{X})$, then $\|\mathcal{K} - \mathcal{K}_M\| \to 0$ as $M \to \infty$. In particular, the eigenvalues of $A_M$ converge to those of $\mathcal{K}$, and $\rho(\mathcal{K}) < 1$ implies $\rho(A_M) < 1$ for all sufficiently large $M$.
\end{proposition}
The completeness requirement is satisfied by standard choices such as the Fourier basis and piecewise-constant (bin) functions.

By standard Kalman filter theory~\cite{Anderson2005}, under detectability of $(A_M, C)$ and stabilizability of $(A_M, \Lambda_v^{1/2})$, the filter covariance $\Psi_{t|t}$ from Lemma~\ref{dgp_lem:seperable_kernels} converges to a unique positive semi-definite limit $\Psi_\infty := \lim_{t\to\infty}\Psi_{t|t}$. When $f_t$ lies exactly in $\operatorname{span}\{u_1, \ldots, u_M\}$ for all $t$, the filter covariance $\Psi_{t|t}$ equals the true error covariance $\mathbb{E}[e_t e_t^\top]$, and the filter is the exact MMSE estimator. When $f_t$ does not lie exactly in the span, $\Psi_{t|t}$ underestimates the true error because the filter is unaware of the leakage from the unmodelled subspace. This gap is quantified in Proposition~\ref{dgp_prop:overconfidence}. Let $\Gamma_\infty$ denote the corresponding steady-state Kalman gain and define
\begin{equation}
\label{dgp_eq:closed_loop}
    \bar{A} := A_M - \Gamma_\infty C.
\end{equation}
The closed-loop matrix $\bar{A}$ is Schur stable, that is, $\rho(\bar{A}) < 1$. This holds regardless of $\rho(\mathcal{K})$. The DGP filter stably estimates functions with unstable dynamics, provided the system is detectable from the available measurement locations.

\subsection{Steady-State Estimation Error for Finite \texorpdfstring{$M$}{M}}
\label{dgp_subsec:finite_m_error}
In general, the function $f_t$ does not lie exactly in $\operatorname{span}\{u_1, \ldots, u_M\}$, so the basis approximation introduces an additional error on top of the noise-induced estimation error. To quantify this effect, let $\pi_M^\perp := I - \pi_M$ and note that the filter estimates only the in-subspace component $\hat{f}_{t|t}(x) = \check{U}^\top(x)\, \hat{z}_{t|t}$. The true pointwise error therefore decomposes as
\begin{equation}
\label{dgp_eq:error_decomp}
    f_t(x) - \hat{f}_{t|t}(x) = \check{U}^\top(x)\, e_t + (\pi_M^\perp f_t)(x),
\end{equation}
where $z_t$ is the coefficient representation of $\pi_M f_t$ as defined in Section~\ref{dgp_subsec:approximate_DGP}, and $e_t := z_t - \hat{z}_{t|t} \in \mathbb{R}^{DM}$ denotes the in-subspace estimation error, i.e., the discrepancy between the true projection coefficients and the filter's estimate. The second term $(\pi_M^\perp f_t)(x)$ is the out-of-subspace residual that no finite-dimensional filter can correct.

Projecting the true dynamics $f_t = \mathcal{K} f_{t-1} + v_t$ onto the basis yields
\begin{equation}
\label{dgp_eq:true_coeff_dynamics}
    z_t = A_M z_{t-1} + \ell_t + v_t^M,
\end{equation}
where $v_t^M$ has covariance $\Lambda_v$ and $\ell_t := \pi_M \mathcal{K}\,\pi_M^\perp f_{t-1}$ is the \emph{leakage}: the contribution of the out-of-subspace part of $f_{t-1}$ into the in-subspace dynamics via the evolution kernel. The filter assumes $\ell_t \equiv 0$, which is not the case when the basis projection is inexact. The filter update for the coefficient vector is (cf.\ Lemma~\ref{dgp_lem:prior_mean})
\begin{equation}
\label{dgp_eq:filter_coeff_update}
    \hat{z}_{t|t} = \hat{z}_{t|t-1} + \Gamma_\infty(Y_t - C\,\hat{z}_{t|t-1}), \qquad \hat{z}_{t+1|t} = A_M\, \hat{z}_{t|t}.
\end{equation}
Subtracting~\eqref{dgp_eq:filter_coeff_update} from~\eqref{dgp_eq:true_coeff_dynamics} gives the error dynamics

\begin{equation}
\label{dgp_eq:error_dynamics}
    e_t = \bar{A}\, e_{t-1} + \ell_t - \Gamma_\infty w_t + v_t^M.
\end{equation}
The leakage term $\ell_t$ depends on the out-of-subspace component of $f_{t-1}$, whose second moments may or may not converge depending on the state dynamics. The following two results formalize this.
\begin{proposition}
\label{dgp_prop:leakage_existence}
If the state dynamics are stable in the sense of Definition~\ref{dgp_def:stability} and the filter is stable ($\rho(\bar{A}) < 1$), then $f_t$ and $e_t$ have bounded second moments, and the following limits exist:
\begin{align}
\label{dgp_eq:WM}
    \Omega_M &:= \lim_{t\to\infty}\mathbb{E}[\ell_t \ell_t^\top] \succeq 0, \\
\label{dgp_eq:cross_cov}
    \Xi_M &:= \lim_{t\to\infty}\mathbb{E}[e_{t-1}\,\ell_t^\top].
\end{align}
The cross-covariance $\Xi_M$ is in general nonzero because $e_{t-1}$ and $\ell_t$ both depend on $f_{t-1}$.
\end{proposition}
\begin{proof}
Stability of the state dynamics ($\rho(\mathcal{K}) < 1$) ensures that $f_t$ converges to a stationary distribution; for the discrete-time case with Hilbert--Schmidt operators, this follows from~\cite[Chapter~11]{DaPrato2014}.
Since $\ell_t = \pi_M \mathcal{K}\,\pi_M^\perp f_{t-1}$ is a bounded linear transformation of $f_{t-1}$, boundedness of $\|\mathbb{E}[\ell_t \ell_t^\top]\| \leq \|\pi_M \mathcal{K}\,\pi_M^\perp\|^2 \cdot \mathbb{E}[\|f_{t-1}\|^2]$ ensures that the covariance inherits this convergence.
For the cross-covariance, note that $e_t$ is driven by $\ell_t$, $w_t$, and $v_t^M$ through the stable recursion~\eqref{dgp_eq:error_dynamics}, so its second moments also converge. Convergence of $\Xi_M$ then follows from the Cauchy--Schwarz inequality: $\|\mathbb{E}[e_{t-1}\ell_t^\top]\| \leq (\mathbb{E}[\|e_{t-1}\|^2])^{1/2}(\mathbb{E}[\|\ell_t\|^2])^{1/2}$, where both factors converge.
\end{proof}

\begin{proposition}
\label{dgp_prop:lyapunov}
Let $(A_M, C)$ be detectable and $(A_M, \Lambda_v^{1/2})$ stabilizable, so that $\rho(\bar{A}) < 1$, and suppose $\Omega_M$ and $\Xi_M$ defined in~\eqref{dgp_eq:WM}--\eqref{dgp_eq:cross_cov} are finite. Then the true steady-state in-subspace error covariance $P_\infty := \lim_{t\to\infty}\mathbb{E}[e_t e_t^\top]$ satisfies the modified discrete Lyapunov equation
\begin{equation}
\label{dgp_eq:modified_lyapunov}
    P_\infty = \bar{A}\, P_\infty\, \bar{A}^\top + \Gamma_\infty W \Gamma_\infty^\top + \Lambda_v + \Omega_M + \bar{A}\,\Xi_M + \Xi_M^\top\bar{A}^\top.
\end{equation}
The steady-state filter covariance $\Psi_\infty$ satisfies
\begin{equation}
\label{dgp_eq:filter_lyapunov}
    \Psi_\infty = \bar{A}\, \Psi_\infty\, \bar{A}^\top + \Gamma_\infty W \Gamma_\infty^\top + \Lambda_v,
\end{equation}
which is~\eqref{dgp_eq:modified_lyapunov} without the leakage terms $\Omega_M$ and $\Xi_M$.
\end{proposition}
\begin{proof}
Recall the error dynamics~\eqref{dgp_eq:error_dynamics}: $e_t = \bar{A}\, e_{t-1} + \ell_t - \Gamma_\infty w_t + v_t^M$. Taking the steady-state second moment and expanding yields
\begin{equation*}
\begin{aligned}
    P_\infty 
    = \bar{A}\, P_\infty\, \bar{A}^\top + \mathbb{E}[\ell_t \ell_t^\top] + \Gamma_\infty W \Gamma_\infty^\top + \Lambda_v 
    + \bar{A}\,\mathbb{E}[e_{t-1}\,\ell_t^\top] + \mathbb{E}[\ell_t\, e_{t-1}^\top]\,\bar{A}^\top,
\end{aligned}
\end{equation*}
where we used that $w_t$ and $v_t^M$ are independent of $(e_{t-1},\,\ell_t)$, so all cross terms with these vanish, and $w_t$ and $v_t^M$ are mutually independent. Identifying $\Omega_M = \mathbb{E}[\ell_t \ell_t^\top]$ and $\Xi_M = \mathbb{E}[e_{t-1}\,\ell_t^\top]$ from Proposition~\ref{dgp_prop:leakage_existence} gives~\eqref{dgp_eq:modified_lyapunov}.

The cross-covariance $\Xi_M$ is generically nonzero: both $e_{t-1} = z_{t-1} - \hat{z}_{t-1|t-1}$ and $\ell_t = \pi_M\mathcal{K}\,\pi_M^\perp f_{t-1}$ are functions of the same realization $f_{t-1}$, so they are correlated. The cross term reflects that the noise driving $e_t$ is colored when the in-subspace and out-of-subspace parts of $f_{t-1}$ are correlated. The filter's Lyapunov equation~\eqref{dgp_eq:filter_lyapunov} follows because the filter assumes $\ell_t \equiv 0$, which eliminates $\Omega_M$ and $\Xi_M$. Proposition~\ref{dgp_prop:leakage_existence} gives sufficient conditions for $\Omega_M$ and $\Xi_M$ to be finite. However, Proposition~\ref{dgp_prop:lyapunov} itself does not require $\rho(\mathcal{K}) < 1$.
\end{proof}
The quantities $P_\infty$ and $\Psi_\infty$ characterize the in-subspace error $e_t$ as a matrix in $\mathbb{R}^{DM \times DM}$. In practice, the quantity of interest is the functional estimation error $\|f_t - \hat{f}_{t|t}\|_{L_2}$. The following result connects the two.
\begin{proposition}
\label{dgp_prop:functional_error}
For all $t \geq 0$,
\begin{equation}
\label{dgp_eq:functional_error}
\begin{aligned}
    \mathbb{E}\!\left[\|f_t - \hat{f}_{t|t}\|_{L_2}^2\right] 
    = \operatorname{tr}\!\left((I_D \otimes \Lambda_U)\, \mathbb{E}[e_t\, e_t^\top]\right) 
    + \mathbb{E}\!\left[\|\pi_M^\perp f_t\|_{L_2}^2\right],
\end{aligned}
\end{equation}
with $\Lambda_U$ the Gram matrix of the basis defined in~\eqref{dgp_eq:KF_matrices}. In particular, as $t \to \infty$ under the conditions of Proposition~\ref{dgp_prop:lyapunov},
\begin{equation}
\label{dgp_eq:functional_error_ss}
\begin{aligned}
    \lim_{t\to\infty}\mathbb{E}\!\left[\|f_t - \hat{f}_{t|t}\|_{L_2}^2\right] 
    = \operatorname{tr}\!\left((I_D \otimes \Lambda_U)\, P_\infty\right) 
     + \lim_{t\to\infty}\mathbb{E}\!\left[\|\pi_M^\perp f_t\|_{L_2}^2\right].
\end{aligned}
\end{equation}
\end{proposition}
\begin{proof}
By~\eqref{dgp_eq:error_decomp}, the error decomposes as $f_t(x) - \hat{f}_{t|t}(x) = \check{U}^\top(x)\, e_t + (\pi_M^\perp f_t)(x)$. The first term lies in $\operatorname{span}\{u_1, \ldots, u_M\}^D$ and the second in its orthogonal complement, so $\langle \check{U}^\top(\cdot)\,e_t,\, \pi_M^\perp f_t\rangle_{L_2} = 0$ for every realization. By the Pythagorean theorem,
\begin{equation*}
    \|f_t - \hat{f}_{t|t}\|_{L_2}^2 = \|\check{U}^\top(\cdot)\, e_t\|_{L_2}^2 + \|\pi_M^\perp f_t\|_{L_2}^2.
\end{equation*}
Computing the first term:
\begin{equation*}
\begin{aligned}
    \|\check{U}^\top(\cdot)\, e_t\|_{L_2}^2 &= e_t^\top\!\left(\int_\mathcal{X} \check{U}(x)\,\check{U}^\top(x)\,d\nu(x)\right) e_t \\
    &= e_t^\top(I_D \otimes \Lambda_U)\, e_t.
\end{aligned}
\end{equation*}
Taking expectations and using $\mathbb{E}[e_t^\top A\, e_t] = \operatorname{tr}(A\,\mathbb{E}[e_t\, e_t^\top])$ gives~\eqref{dgp_eq:functional_error}.
\end{proof}
The first term in~\eqref{dgp_eq:functional_error} is the in-subspace estimation error, weighted by the Gram matrix of the basis functions. The second term is the out-of-subspace residual that no finite-dimensional filter can correct. For an orthonormal basis ($\Lambda_U = I_M$), the first term reduces to $\operatorname{tr}(\mathbb{E}[e_t\, e_t^\top])$.

The filter reports $\Psi_{t|t}$ as its error covariance, but this differs from $\mathbb{E}[e_t\, e_t^\top]$ when the basis projection is inexact. The following proposition quantifies the gap between $P_\infty$ and $\Psi_\infty$.
\begin{proposition}
\label{dgp_prop:overconfidence}
Let $(A_M, C)$ be detectable and $(A_M, \Lambda_v^{1/2})$ stabilizable, and suppose $\Omega_M$ and $\Xi_M$ exist. Then $P_\infty - \Psi_\infty$ satisfies the discrete Lyapunov equation
\begin{equation}
\label{dgp_eq:overconfidence}
    P_\infty - \Psi_\infty = \bar{A}(P_\infty - \Psi_\infty)\bar{A}^\top + \Omega_M + \bar{A}\,\Xi_M + \Xi_M^\top\bar{A}^\top.
\end{equation}
In particular, $P_\infty \succeq \Psi_\infty$ whenever $\Omega_M + \bar{A}\,\Xi_M + \Xi_M^\top\bar{A}^\top \succeq 0$, and the gap satisfies
\begin{equation}
\label{dgp_eq:overconfidence_bound}
    \|P_\infty - \Psi_\infty\| \leq \frac{\|\Omega_M + \bar{A}\,\Xi_M + \Xi_M^\top\bar{A}^\top\|}{1 - \|\bar{A}\|^2}.
\end{equation}
\end{proposition}
\begin{proof}
Subtracting the Lyapunov equations for $\Psi_\infty$ and $P_\infty$ gives~\eqref{dgp_eq:overconfidence}. Iterating yields the explicit solution $P_\infty - \Psi_\infty = \sum_{k=0}^\infty \bar{A}^k\, (\Omega_M + \bar{A}\,\Xi_M + \Xi_M^\top\bar{A}^\top)\, (\bar{A}^\top)^k$. If the driving term is positive semidefinite, every summand is positive semidefinite, so $P_\infty \succeq \Psi_\infty$. Taking norms gives the bound.
\end{proof}
This result generalizes the reduced-order Kalman filter analysis of~\cite{Farrell2001} and~\cite{Rozier2007} to the setting where the full model is infinite-dimensional rather than a finite-dimensional truncation. An optimal reduced-order filter in the sense of~\cite{Bernstein1985} would additionally account for $\Omega_M$ and $\Xi_M$ in the gain computation. In the exact DGP case, $\pi_M^\perp f_t \equiv 0$ for all $t$, so $\Omega_M = \Xi_M = 0$ and $P_\infty = \Psi_\infty$.

Combining Propositions~\ref{dgp_prop:functional_error} and~\ref{dgp_prop:overconfidence} gives a three-way decomposition of the steady-state functional error:
\begin{equation}
\label{dgp_eq:functional_error_threeway}
\begin{aligned}
    \lim_{t\to\infty}\mathbb{E}\!\left[\|f_t - \hat{f}_{t|t}\|_{L_2}^2\right]
    &= \underbrace{\operatorname{tr}\!\left((I_D \otimes \Lambda_U)\, \Psi_\infty\right)}_{\text{noise-limited error}} 
    + \underbrace{\operatorname{tr}\!\left((I_D \otimes \Lambda_U)(P_\infty - \Psi_\infty)\right)}_{\text{leakage-induced gap}} \\
    &\qquad
    + \underbrace{\lim_{t\to\infty}\mathbb{E}\!\left[\|\pi_M^\perp f_t\|_{L_2}^2\right]}_{\text{out-of-subspace residual}}.
\end{aligned}
\end{equation}
The first term is the estimation error due to measurement noise and process disturbances alone, and can be computed directly from the filter output. The second term is the additional in-subspace error caused by the filter ignoring the leakage $\ell_t$, and is nonnegative whenever the condition in Proposition~\ref{dgp_prop:overconfidence} holds. The third term is always nonnegative.

\subsection{Convergence as \texorpdfstring{$M \to \infty$}{M to infinity}}
\label{dgp_subsec:convergence}
As the number of basis functions grows, all sources of approximation error vanish, and the approximate DGP recovers the exact infinite-dimensional estimator.
\begin{proposition}
\label{dgp_prop:convergence}
Under the conditions of Propositions~\ref{dgp_prop:galerkin_convergence} and~\ref{dgp_prop:overconfidence}, and assuming that the basis $\{u_m\}_{m=1}^\infty$ is complete in $L_2(\mathcal{X})$, the following limits hold as $M \to \infty$:
\begin{enumerate}
    \item $\Omega_M \to 0$, $\Xi_M \to 0$, and $\|P_\infty - \Psi_\infty\| \to 0$, so that the overconfidence vanishes.
    \item $\|\pi_M^\perp f_t\|_{L_2} \to 0$ for all $t$, so that the out-of-subspace residual vanishes.
    \item The filter covariance $\Psi_\infty$ converges to the steady-state posterior covariance of the exact infinite-dimensional DGP.
    \item All three terms in~\eqref{dgp_eq:functional_error_threeway} converge, and $\lim_{t\to\infty}\mathbb{E}[\|f_t - \hat{f}_{t|t}\|_{L_2}^2]$ converges to the steady-state functional error of the exact DGP.
\end{enumerate}
\end{proposition}
\begin{proof}
The first limit follows from Proposition~\ref{dgp_prop:galerkin_convergence}, which gives $\|\pi_M \mathcal{K}\,\pi_M^\perp\| \to 0$, so that both $\Omega_M \to 0$ and $\Xi_M \to 0$. Combined with the bound in~\eqref{dgp_eq:overconfidence_bound}, this gives $\|P_\infty - \Psi_\infty\| \to 0$. The second limit follows from completeness of the basis. The third limit uses continuity of the DARE solution in the system matrices~\cite{Anderson2005}. The fourth limit follows from Proposition~\ref{dgp_prop:functional_error} and items 1--3.
\end{proof}
The practical implication is that increasing $M$ reduces all error contributions with diminishing returns, which is verified numerically in Section~\ref{dgp_sec:results}.

\section{Numerical Examples}
\label{dgp_sec:results}
This section presents two numerical examples that demonstrate the DGP estimator and its separable kernel approximation from Section~\ref{dgp_sec:analysis}. Each example begins with the derivation of the evolution kernel from the underlying PDE, following the connection described in Section~\ref{dgp_subsec:problem}.

\subsection{Example 1: Heat Equation}
\label{dgp_sec:results_heat}
Consider the one-dimensional heat equation on $\mathcal{X} = \mathbb{R}$,
\begin{equation}
\label{dgp_eq:heat_pde}
    \frac{\partial f}{\partial \tau}(x,\tau) = \alpha \frac{\partial^2 f}{\partial x^2}(x,\tau),
\end{equation}
in which $\alpha > 0$ is the thermal diffusivity. The Green's function of~\eqref{dgp_eq:heat_pde} on $\mathbb{R}$ is~\cite[\S2.3]{Evans2010}
\begin{equation}
\label{dgp_eq:heat_greens}
    G(x,s,\tau) = \frac{1}{\sqrt{4\pi \alpha \tau}} \exp\!\left(-\frac{(x-s)^2}{4\alpha \tau}\right),
\end{equation}
which maps the initial condition at location $s$ to the solution at location $x$ after time $\tau$.

\subsubsection{Evolution Kernel}
\label{dgp_sec:heat}
Defining $f_t(x) := f(x,t\Delta)$ for $t \in \mathbb{N}$ and setting $\tau = \Delta$ in~\eqref{dgp_eq:heat_greens} gives the evolution kernel
\begin{equation}
\label{dgp_eq:heat_kernel}
    k_f(x,s) = \frac{1}{\sqrt{4\pi \alpha \Delta}} \exp\!\left(-\frac{(x-s)^2}{4\alpha \Delta}\right),
\end{equation}
which is a squared exponential kernel with length scale $\sigma_k = \sqrt{2\alpha \Delta}$ and amplitude $a_k = (2\pi\sigma_k^2)^{-1/2}$. This is a scalar ($D=1$) problem, so $K_f$ reduces to a scalar-valued kernel.

\subsubsection{Setup}
Consider the heat equation on the bounded domain $\mathcal{X} = [-1,1]$ with evolution kernel~\eqref{dgp_eq:heat_kernel}. On a bounded domain, the system is naturally strictly stable in the sense of \textit{Definition}~\ref{dgp_def:stability}, since mass that diffuses past the domain boundary is lost.

To assess the estimation accuracy, a ground truth model is generated using 625 discretized basis functions, while the approximate DGP uses up to 101 Fourier bases. The initial condition covariance, process noise covariance, and measurement noise covariance are
\begin{align*}
Q_f(x,x') &= a_f \exp\!\left(-\frac{(x-x')^2}{2\sigma_f^2}\right),
\qquad Q_v(x,x') &= a_v \exp\!\left(-\frac{(x-x')^2}{2\sigma_v^2}\right),\\
Q_w(x,x') &= \sigma_w^2 \,\delta(x-x'),
\end{align*}
with $a_f = 0.1$, $\sigma_f = 0.3$, $a_v = 0.1$, $\sigma_v^2 = 0.1$, and $\sigma_w^2 = 0.1$. The initial condition covariance $Q_f$ ensures that $f_0$ smoothly deviates from its mean. The measurement noise $Q_w$ is spatially white. The initial mean function is
\begin{equation*}
    \bar{f}_0(x) =
    \begin{cases}
        10 & |x|<0.05, \\
        0 & \text{otherwise,}
    \end{cases}
\end{equation*}
which simulates an impulse-like function at the center of the domain. At each time step, $p=5$ observation locations are drawn uniformly from $\mathcal{X}$.

\begin{figure}[!t]
    \centering
    \includegraphics[width=0.98\textwidth]{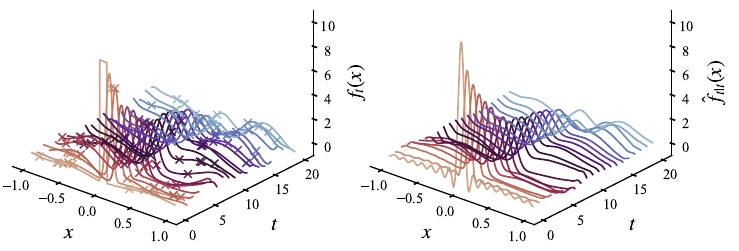}
    \vspace{-0.3cm}
    \caption{Heat equation example. Left: ground truth $f_t(x)$ with randomly placed observation locations ({\color{twilightcenter}$\boldsymbol{\times}$}). Right: DGP estimate $\hat{f}_{t|t}(x)$ using $M=31$ Fourier basis functions. Each line represents one discrete time step; line colour indicates the time index.}
    \label{dgp_fig:heat_estimation}
\end{figure}

Figure~\ref{dgp_fig:heat_estimation} shows the ground truth and the DGP estimate side by side. The smoothing of the function as the time index $t$ increases is characteristic of the heat equation: the initial impulse diffuses and spreads over the spatial domain. The effect of the process disturbances $v_t$ is visible in the ground truth, for instance near $x = -0.5$ at time step $t = 3$. The DGP estimate tracks the evolving function well despite these disturbances. At $t = 0$, however, a noticeable approximation error is visible: the discontinuous initial condition $\bar{f}_0$ cannot be exactly represented by a finite number of smooth Fourier basis functions, and the resulting overshoot is a manifestation of the Gibbs phenomenon. Because the heat kernel smooths the function over time, this representation error vanishes quickly as the ground truth becomes increasingly well-suited to the Fourier basis.

\begin{figure}[!t]
    \centering
    \includegraphics[width=0.8\textwidth]{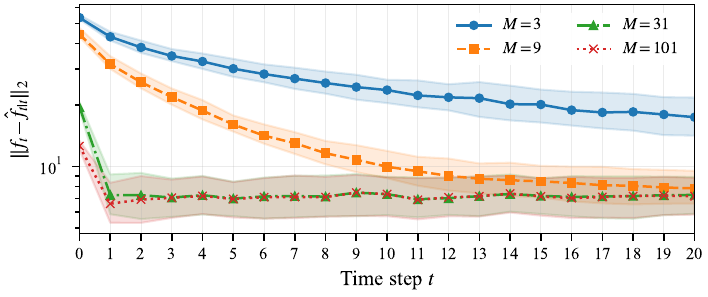}
    \vspace{-0.3cm}
    \caption{$2$-norm of the estimation error $\| f_t - \hat{f}_{t|t} \|_2$ over time steps $t$ for the heat equation example, comparing $M \in \{3, 9, 31, 101\}$ Fourier bases. Lines show the median over $500$ Monte Carlo runs. Shaded regions indicate the 25th--75th percentile range.}
    \label{dgp_fig:heat_convergence}
\end{figure}

Figure~\ref{dgp_fig:heat_convergence} shows the $2$-norm of the estimation error for different numbers of basis functions $M$. Increasing $M$ reduces the estimation error with diminishing returns, which is consistent with Proposition~\ref{dgp_prop:convergence}: as $M$ increases, all three error contributions in~\eqref{dgp_eq:functional_error_threeway} decrease. The error also decreases over time for all choices of $M$, which can be explained by the smoothing effect of the heat kernel: as the true function evolves, it becomes easier to represent in the Fourier basis. This effect is least visible for $M=101$, which can already represent non-smooth functions such as the initial impulse.

\subsection{Example 2: Wave Equation}
\label{dgp_sec:results_wave}

Consider the one-dimensional wave equation on $\mathcal{X} = \mathbb{R}$,
\begin{equation}
\label{dgp_eq:wave_pde}
    \frac{\partial^2 \varphi}{\partial \tau^2}(x,\tau) = c^2 \frac{\partial^2 \varphi}{\partial x^2}(x,\tau),
\end{equation}
in which $c > 0$ is the wave propagation speed. The wave equation is second-order in time, and its solution depends on two initial conditions: the initial position $\varphi(x,0) = \varphi_0(x)$ and the initial velocity $\psi(x,0) := \frac{\partial \varphi}{\partial \tau}(x,0) = \psi_0(x)$. By D'Alembert's formula, the solution is
\begin{equation}
\label{dgp_eq:dalembert}
\begin{aligned}
    \varphi(x,\tau) 
    = \frac{1}{2}\!\left[\varphi_0(x - c\tau) + \varphi_0(x + c\tau)\right]
    + \frac{1}{2c}\int_{x-c\tau}^{x+c\tau} \psi_0(s)\, ds.
\end{aligned}
\end{equation}

\subsubsection{Evolution Kernels}
\label{dgp_sec:wave}
Because the wave equation is second-order in time, we define $\varphi_t(x) := \varphi(x,t\Delta)$ and $\psi_t(x) := \frac{\partial \varphi}{\partial \tau}(x,t\Delta)$ for $t \in \mathbb{N}$, and set $f_t(x) := [\varphi_t(x),\, \psi_t(x)]^\top \in \mathbb{R}^2$. Substituting the state at discrete time $t$ as the initial condition and evaluating D'Alembert's formula~\eqref{dgp_eq:dalembert} at $\tau = \Delta$ gives the position update
\begin{equation}
\label{dgp_eq:wave_f_update}
\begin{aligned}
    \varphi_{t+1}(x) 
    = \frac{1}{2}\!\left[\varphi_t(x - c\Delta) + \varphi_t(x + c\Delta)\right]
    + \frac{1}{2c}\int_{x-c\Delta}^{x+c\Delta} \psi_t(s)\, ds.
\end{aligned}
\end{equation}
The velocity update is obtained by differentiating~\eqref{dgp_eq:dalembert} with respect to time. By the chain rule, $\frac{\partial}{\partial \tau}\varphi_t(x \pm c\tau)\big|_{\tau=\Delta}$ produces a factor $\pm c$ and converts the time derivative into a spatial derivative $\varphi_t'$, giving
\begin{equation}
\label{dgp_eq:wave_g_update}
\begin{aligned}
    \psi_{t+1}(x) 
    = \frac{c}{2}\!\left[\varphi_t'(x + c\Delta) - \varphi_t'(x - c\Delta)\right]
    + \frac{1}{2}\!\left[\psi_t(x - c\Delta) + \psi_t(x + c\Delta)\right],
\end{aligned}
\end{equation}
where $\varphi_t'(\cdot) := \frac{d}{d(\cdot)} \varphi_t(\cdot)$ denotes the spatial derivative.

Following~\eqref{dgp_eq:general_solution}, the evolution kernels are identified by rewriting~\eqref{dgp_eq:wave_f_update}--\eqref{dgp_eq:wave_g_update} as integrals against the state. From~\eqref{dgp_eq:wave_f_update}, the shifted evaluations of $\varphi_t$ and the integral over $\psi_t$ give
\begin{equation}
\label{dgp_eq:wave_kernels_exact}
\begin{aligned}
    k_{\varphi\varphi}(x,s) &= \tfrac{1}{2}\!\left[\delta(x\!-\!s\!-\!c\Delta) + \delta(x\!-\!s\!+\!c\Delta)\right],\\
    k_{\varphi\psi}(x,s) &= \tfrac{1}{2c}\,\mathbf{1}_{[x-c\Delta,\, x+c\Delta]}(s),
\end{aligned}
\end{equation}
where $\mathbf{1}_{[a,b]}(s)$ denotes the indicator function on $[a,b]$. From~\eqref{dgp_eq:wave_g_update}, the spatial derivatives of $\varphi_t$ yield
\begin{equation}
\label{dgp_eq:wave_kernels_exact2}
\begin{aligned}
    k_{\psi\varphi}(x,s) &= \tfrac{c}{2}\!\left[\delta'(x\!-\!s\!+\!c\Delta) - \delta'(x\!-\!s\!-\!c\Delta)\right],\\
    k_{\psi\psi}(x,s) &= k_{\varphi\varphi}(x,s),
\end{aligned}
\end{equation}
where $\delta'$ denotes the distributional derivative of the Dirac delta. The symmetry $k_{\psi\psi} = k_{\varphi\varphi}$ is a direct consequence of D'Alembert's formula. With these kernels, the updates~\eqref{dgp_eq:wave_f_update}--\eqref{dgp_eq:wave_g_update} take the form of a coupled IDE:
\begin{equation}
\label{dgp_eq:wave_ide_system}
    \underbrace{\begin{bmatrix}
        \varphi_{t+1}(x) \\ \psi_{t+1}(x)
    \end{bmatrix}}_{f_{t+1}(x)}
    =
    \int_\mathcal{X}
    \underbrace{\begin{bmatrix}
        k_{\varphi\varphi}(x,s) & k_{\varphi\psi}(x,s) \\
        k_{\psi\varphi}(x,s) & k_{\psi\psi}(x,s)
    \end{bmatrix}}_{K_f(x,s)\,\in\,\mathbb{R}^{2\times 2}}
    \underbrace{\begin{bmatrix}
        \varphi_t(s) \\ \psi_t(s)
    \end{bmatrix}}_{f_t(s)} ds,
\end{equation}
which is an instance of~\eqref{dgp_eq:dynamic_function_system} with $D=2$. Because the kernels in~\eqref{dgp_eq:wave_kernels_exact}--\eqref{dgp_eq:wave_kernels_exact2} are distributions rather than ordinary functions, they are replaced by Gaussian approximations with smoothing parameter $\epsilon > 0$ for the numerical implementation. The wave equation conserves energy, so $\rho(\mathcal{K}) = 1$ and the steady-state results of Propositions~\ref{dgp_prop:leakage_existence}--\ref{dgp_prop:convergence} do not formally apply. However, the closed-loop matrix $\bar{A}$ from~\eqref{dgp_eq:closed_loop} is Schur stable regardless of whether the open-loop system is stable, so the filter can track a bounded signal even when the open-loop dynamics are marginally stable. In the disturbance-free setting below, the ground truth is the bounded analytical solution, and the estimation error remains bounded because the filter stably incorporates new observations at each time step.

\subsubsection{Setup}
The example is evaluated on the bounded domain $x \in [-10, 10]$ with wave speed $c = 2$ and time step $\Delta = 0.2$. The state is $f_t(x) = [\varphi_t(x),\, \psi_t(x)]^\top$, where $\varphi_t$ is the position function and $\psi_t$ is the velocity function. The distributional evolution kernels are replaced by Gaussian approximations with smoothing parameter $\epsilon = h$, where $h > 0$ is the spatial grid spacing. As $\epsilon \to 0$, these smoothed kernels converge to the exact distributional kernels.

The ground truth is the analytical D'Alembert solution $\varphi(x,\tau) = \tfrac{1}{2}[m_0(x - c\tau) + m_0(x+c\tau)]$ with initial position $m_0(x) = 10\exp(-x^2/2)$ and initial velocity $\psi_0(x) = 0$. This produces two counter-propagating Gaussian pulses. Because this solution is deterministic, there are no process disturbances. The DGP estimator uses $M=31$ Fourier basis functions per state component (so $2M = 62$ coefficients in total) with $p=3$ randomly placed measurements at each time step and measurement noise variance $\sigma_w^2 = 10^{-5}$. The prior mean is set to zero, so that the estimator must discover the propagating wavefronts entirely from the observations.

Unlike the heat equation, which diffuses energy over the domain, the wave equation preserves the shape of propagating wavefronts. This makes the estimation problem fundamentally different: the evolution kernel is not smoothing, and the function does not become easier to approximate over time.

\begin{figure}[!t]
    \centering
    \includegraphics[width=0.98\textwidth]{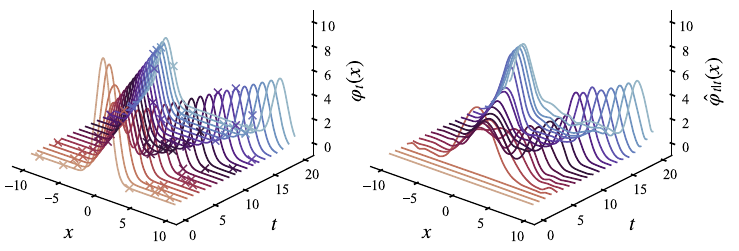}
    \vspace{-0.3cm}
    \caption{Wave equation example. Left: analytical D'Alembert solution $\varphi_t(x)$ with observation locations ({\color{twilightcenter}$\boldsymbol{\times}$}). Right: DGP estimate $\hat{\varphi}_{t|t}(x)$ using $M=31$ Fourier bases. Each line represents one discrete time step; line colour indicates the time index.}
    \label{dgp_fig:wave_estimation}
\end{figure}

Figure~\ref{dgp_fig:wave_estimation} shows the ground truth and the DGP estimate for the wave equation. The two counter-propagating pulses are clearly visible in the analytical solution. Despite starting from a zero prior mean and receiving only $p=3$ noisy measurements per time step, the DGP estimator recovers the shape of both propagating wavefronts within a few time steps. Because the wave equation does not smooth the solution, the Fourier basis must resolve the sharp features of the pulse at every time step, in contrast to the heat equation example where the kernel itself reduces the required bandwidth over time.

\section{Conclusions}
\label{dgp_sec:conclusions}
This paper presented the Dynamic Gaussian Process (DGP), which unifies Gaussian process regression and Kalman filtering for the estimation of functions governed by integro-difference equations. The stability and approximation error analysis showed that the functional estimation error decomposes exactly into three distinct components. Specifically, this error comprises the noise-limited error, the leakage-induced gap, and the out-of-subspace residual, all of which vanish as $M \to \infty$. An application of the estimation method to the heat equation example confirms this convergence, while an application to the wave equation demonstrated that the filter can track vector-valued functions.

Several directions for future research remain. Deriving explicit rates for the leakage covariance $\|\Omega_M\|$ as a function of $M$, the domain $\mathcal{X}$, the kernel $K_f$, and the stability margin would yield practical guidelines for choosing $M$. Extensions to generalize the model with control inputs or to include nonlinear PDEs are of interest. The posterior covariance of the DGP can also be used to determine where to place the next observation, enabling frameworks for sensor scheduling and Bayesian optimization.

\bibliographystyle{siamplain}
\bibliography{references}

\end{document}